\documentclass[12pt]{amsart}
\usepackage[utf8]{inputenc}
\usepackage{amssymb, amsmath}
\usepackage{fourier}

\theoremstyle{plain}

\newtheorem{theorem}{Theorem}

\newtheorem{lemma}[theorem]{Lemma}
\newtheorem{claim}[theorem]{Claim}

\newtheorem{definition}[theorem]{Definition}

\newtheorem{proposition}[theorem]{Proposition}

\theoremstyle{remark}
\newtheorem*{Remark}{Remark}

\newtheorem{remark}[theorem]{Remark}

\newcommand{\bit}{\begin{itemize}}
\newcommand{\eit}{\end{itemize}}
\newcommand{\ben}{\begin{enumerate}}
\newcommand{\een}{\end{enumerate}}
\newcommand{\be}{\begin{equation}}
\newcommand{\ee}{\end{equation}}
\newcommand{\ba}{\begin{array}}
\newcommand{\ea}{\end{array}}

\newcommand{\mc}[1]{\mathcal{#1}}
\newcommand{\supp}[1]{\mathrm{supp}\left(#1\right)}

\newcommand{\norm}[1]{\left|\left|#1\right|\right|}

\newcommand{\eps}{\varepsilon}

\newcommand{\inner}[2]{\left< #1,#2 \right>}

\newcommand\cH{\mathcal H}
\newcommand\cM{\mathcal M}
\newcommand\cP{\mathcal P}
\newcommand\cS{\mathcal S}

\newcommand\R{\mathbb R}
\newcommand\N{\mathbb N}

\newcommand\rng{\operatorname{rng}}
\newcommand\tr{\operatorname{Tr}}
\newcommand{\ler}[1]{\left( #1 \right)}

\begin{document}

\title[Maps preserving the quantum $\chi_\alpha^2$-divergence]{Maps on positive definite operators preserving the quantum $\chi_\alpha^2$-divergence}

\author[H. Y. Chen]{Hong-Yi Chen}
\address{Department of Applied Mathematics\\
National Sun Yat-sen University\\
Kaohsiung 80424, Taiwan.}
\email{hongyi0906@gmail.com}

\author[Gy. P. Geh\'er]{Gy\"orgy P\'al Geh\'er}
\address{MTA-SZTE Analysis and Stochastics Research Group,
Bolyai Institute, University of Szeged,
H-6720 Szeged, Aradi v\'ertan\'uk tere 1., Hungary
and
MTA-DE ``Lend\" ulet'' Functional Analysis Research Group, Institute of Mathematics\\
         University of Debrecen\\
         H-4002 Debrecen, P.O. Box 400, Hungary}
\email{gehergy@math.u-szeged.hu}
\urladdr{http://www.math.unideb.hu/\~{}gehergy}

\author[C. N. Liu]{Chih-Neng Liu}
\address{Department of Applied Mathematics\\
National Sun Yat-sen University\\
Kaohsiung 80424, Taiwan.}
\email{cnliu@mail.nsysu.edu.tw}

\author[L. Moln\'ar]{Lajos Moln\'ar}
\address{Bolyai Institute\\
University of Szeged\\
H-6720 Szeged, Aradi v\'ertan\'uk tere 1.,
Hungary and Institute of Mathematics\\
Budapest University of Technology and Economics\\
H-1521 Budapest, Hungary and
MTA-DE ``Lend\" ulet'' Functional Analysis Research Group, Institute of Mathematics\\
         University of Debrecen\\
         H-4002 Debrecen, P.O. Box 400, Hungary}
\email{molnarl@math.u-szeged.hu}
\urladdr{http://www.math.u-szeged.hu/\~{}molnarl}

\author[D. Virosztek]{D\'aniel Virosztek}
\address{Institute of Mathematics\\
Budapest University of Technology and Economics\\
H-1521 Budapest, Hungary
and
MTA-DE ``Lend\" ulet'' Functional Analysis Research Group, Institute of Mathematics\\
         University of Debrecen\\
         H-4002 Debrecen, P.O. Box 400, Hungary}
\email{virosz@math.bme.hu}
\urladdr{http://www.math.bme.hu/\~{}virosz}

\author[N. C. Wong]{Ngai-Ching Wong}
\address{Department of Applied Mathematics\\
National Sun Yat-sen University\\
Kaohsiung 80424, Taiwan.}
\email{wong@math.nsysu.edu.tw}
\urladdr{http://www.math.nsysu.edu.tw/~wong}

\thanks{
Gy. P. Geh\'er and L. Moln\'ar were supported by the ``Lend\"ulet'' Program (LP2012-46/2012) of the Hungarian Academy of Sciences and by the National Research, Development and Innovation Office – NKFIH, Grant No. K115383.
D. Virosztek was supported by the ``Lend\"ulet'' Program (LP2012-46/2012) of the Hungarian Academy of Sciences, by the National Research, Development and Innovation Office – NKFIH, Grant No. K104206, and by the ``For the Young Talents of the Nation'' scholarship program (NTP-EF\"O-P-15-0481) of the Hungarian State. The project was also supported by the joint venture of Taiwan and Hungary MOST-HAS, Grant No. 104-2911-1-110-508.}



\begin{abstract}
We describe the structure of all bijective maps on the cone of positive definite operators acting on a finite and at least two-dimensional complex Hilbert space which preserve the quantum $\chi_\alpha^2$-divergence for some $\alpha \in [0,1]$. We prove that any such transformation is necessarily implemented by either a unitary or an antiunitary operator. Similar results concerning maps on the cone of positive semidefinite operators as well as on the set of all density operators are also derived.
\end{abstract}
\maketitle
\section{Introduction}

The study of automorphisms, symmetries or, more generally, maps on mathematical structures which preserve relevant characteristics (numerical or nonnumerical) of the underlying structures is an important general task in most areas of mathematics and its applications, hence in mathematical physics, too. In the latter discipline  one of the most fundamental corresponding result is Wigner's celebrated theorem on the structure of so-called quantum mechanical symmetry transformations. These transformations are bijective maps on the set of all rank-one projections on a complex Hilbert space (representing the pure states of a quantum system) that preserve the quantity of transition probability which is the trace of the product of rank-one projections. Wigner's theorem states that any such map is implemented by a unitary or antiunitary operator on the underlying Hilbert space. Motivated by this very important result, in a series of papers we presented several results in which we determined the structures of transformations on the sets of density operators or  positive (definite or semidefinite) operators that preserve certain kinds of quantum divergence. Below we list those results of ours which are in close connections to the present investigations. 

In order to do this, let us first fix the notation. In what follows $\mc{H}$ stands for a finite and at least two-dimensional complex Hilbert space, $d=\dim \mc{H}$, and we denote by $\mc{L}(\mc{H})$  the set of all linear operators on $\mc{H}.$  The symbols $\mc{L}^{sa}(\mc{H}), \mc{L}^{+}(\mc{H})$ and $\mc{L}^{++}(\mc{H})$ stand for the collections of all selfadjoint, positive semidefinite, and positive definite operators on $\cH,$ respectively. The linear space $\mc{L}(\mc{H})$ is endowed with the Hilbert-Schmidt inner product $\inner{X}{Y}_{HS}=\tr X Y^{*}$, $X,Y\in \mc{L}(\mc{H})$, and $\norm{.}_{HS}$ denotes the induced norm. We will also consider the operator norm on $\mc{L}(\mc{H})$ which is denoted by $\norm{.}_{op}.$ The symbol $\cS\ler{\cH}$ stands for the set of all density operators on $\cH$, i.e., the set of all elements in $\mc{L}^{+}(\mc{H})$ with unit trace. The elements of $\cS\ler{\cH}$ represent the quantum states of the quantum system described by the Hilbert space $\cH$, hence $\cS\ler{\cH}$ is also called state space. The set of all nonsingular (i.e., invertible) elements of $\cS\ler{\cH}$ is denoted by $\cM\ler{\cH}$ and $\cP_1(\cH)$ stands for the set of all rank-one projections on $\cH.$

If $f: \, I \rightarrow \R$ is a function defined on an interval $I \subset \R$, then the corresponding \emph{standard operator function} is the map
$$
f: \{A \in \mc{L}^{sa}(\mc{H}): \ \sigma(A) \subseteq I \} \rightarrow \mc{L}(\mc{H})
$$
$$
A=\sum_{a \in \sigma(A)} a P_a \longmapsto f(A):=\sum_{a \in \sigma(A)} f(a) P_a,
$$
where $\sigma(A)$ is the spectrum of $A$ and $P_a$ is the spectral projection corresponding to the eigenvalue $a$ of $A$.
\par
Wigner's above mentioned fundamental theorem states that any bijective map
$\phi:\cP_1(\cH)\to \cP_1(\cH)$ which has the property that
\[
\tr \phi(P)\phi(Q)=\tr PQ \qquad (P,Q\in \cP_1(\cH))
\]
is necessarily of the form
\begin{equation}\label{E:Wig}
\phi(P)=UPU^* \qquad \ler{P\in \cP_1(\cH)}
\end{equation}
with some either unitary or antiunitary operator $U$ on $\cH$. (There is a vast literature on this celebrated result, we only refer to Sections 0.3, 2.1 and 2.2 in the monograph \cite{MB} and to the recent elementary proof given in \cite{gyuri}.)

And now a short summary of our former and relating results follows. We begin with noting that divergences, in particular, relative entropy type quantities are usually defined on the state space or on the cones of positive definite or semidefinite operators depending on the nature of the problem one considers. Therefore, we investigated  the corresponding preserver transformations on all those structures. Obviously, the machinery we used in our arguments to obtain the results heavily depended on which particular structures the maps were defined. 
  
In the paper \cite{molszok} we proved that those (a priori nonbijective) maps on the state space $\cS\ler{\cH}$ which preserve the (Umegaki) relative entropy have the structure like in Wigner's theorem \eqref{E:Wig}, they are all implemented by unitary or antiunitary operators. 
Next, in \cite{mnsz13} we presented a far reaching generalization of the result in \cite{molszok} by showing  that all maps on $\mc{S}(\mc{H})$ which preserve a so-called $f$-divergence ($f$ being an arbitrary strictly convex real function on the set of nonnegative real numbers) are also unitary or antiunitary similarity transformations.
In \cite{dv16b} the same conclusion was obtained for the same kind of preservers which are bijective and defined not on the state space but on the whole set $\mc{L}^+(\mc{H})$ of positive semidefinite operators. (We also remark that in the very recent paper \cite{ML16g} we have made some steps towards the description of quasi-entropy preservers on positive definite cones in the setting of $C^*$-algebras but the level of generality of the considered quasi-entropies falls far from what we could consider sufficient.)  

In \cite{mpv15} we described the structure of all bijective maps on the positive definite cone $\mc{L}^{++}(\mc{H})$ which preserve the Bregman divergence corresponding to any differentiable convex function on the positive reals with derivative bounded from below and unbounded from above. In addition, we considered the cases of the particular functions $x\mapsto x\log x$, $x>0$ (von Neumann divergence, in other words, Umegaki's relative entropy) and that of $x\mapsto -\log x$ (Stein's loss). In the former cases the preservers are all unitary-antiunitary conjugations while in the latter one they are conjugations by any invertible linear or conjugate-linear operators on $\mc{H}$. In the same paper we obtained results of similar spirit concerning maps on $\mc{L}^{++}(\mc{H})$ preserving Jensen divergence. Similar investigation was carried out in \cite{dv16a} for bijective transformations on the state space preserving Bregman or Jensen divergences.

In the present paper we consider a relatively new and important notion of quantum divergence and determine its preservers. We emphasize in advance that the problem has been a real challenge, the formerly developed techniques have needed to be altered significantly and many new ideas have necessarily had to be brought in. Now, the basic concept of the present paper is the following notion of quantum divergence which was introduced in \cite{Ruskaietal}, see equation (7) on page 122201-3. 

\begin{definition}\label{D:1}
Let $\alpha \in [0,1].$ The quantum $\chi_\alpha^2$-divergence of the operators $A \in \mc{L}^+(\cH)$, $B\in \mc{L}^{++}(\cH)$ is defined by
\begin{equation*} 
K_\alpha(A||B):=\tr B^{-\alpha}(A-B)B^{\alpha-1}(A-B).
\end{equation*}
Clearly, we can also write this as
\begin{equation*}\label{E:M1}
K_\alpha(A||B)=\tr B^{-\alpha}AB^{\alpha-1}A-2\tr A+\tr B.
\end{equation*}
For a singular $B\in \mc{L}^+(\cH)$ we define
\begin{equation} \label{def1}
K_\alpha(A||B):=\lim_{\eps\to 0} K_\alpha(A||B+\eps I).
\end{equation}
\end{definition}

\begin{remark}\label{R:M}
In relation with the above definition we make a few comments. First, concerning the existence of the limit in \eqref{def1} observe the following. In the case of a singular $B\in \mc{L}^+(\cH)$
one can easily see that if $\supp{A} \subseteq \supp{B}$ ($\supp{B}$ denoting the support of $B$ which is the orthogonal complement of its kernel hence equals the range $\rng{(B)}$ of $B$), then we have $K_\alpha(A||B)=\tr B^{-\alpha}(A-B)B^{\alpha-1}(A-B)$, where the trace is taken over the subspace $\supp{B}$ of $\mc{H}$. 

If $\supp{A} \nsubseteq \supp{B}$, then we have $K_\alpha(A||B)=\infty$. Indeed, assume that the sequence $\{\tr (B+\epsilon_n I)^{-\alpha}A(B+\epsilon_n I)^{\alpha-1}A\}_{n \in \N}$ is bounded for some sequence $\{\epsilon_n\}_{n \in \N}$ of positive numbers converging to zero. Since we have 
\[
\tr (B+\epsilon_n I)^{-\alpha}A(B+\epsilon_n I)^{\alpha-1}A=\tr \left|(B+\epsilon_n I)^\frac{\alpha -1}{2}A(B+\epsilon_n I)^\frac{-\alpha}{2}\right|^2,
\]
this yields that $\{(B+\epsilon_n I)^\frac{\alpha -1}{2}A(B+\epsilon_n I)^\frac{-\alpha}{2}\}_{n \in \N}$ is a bounded sequence in the Hilbert-Schmidt norm and hence it has a convergent subsequence. Without serious loss of generality we may and do assume that already the original sequence itself converges 
\[
(B+\epsilon_n I)^\frac{\alpha -1}{2}A(B+\epsilon_n I)^\frac{-\alpha}{2} \to C.
\]
Since
\[
(B+\epsilon_n I)^\frac{1-\alpha}{2}\to B^\frac{1-\alpha}{2}, \quad  (B+\epsilon_n I)^\frac{\alpha}{2}\to B^\frac{\alpha}{2},
\]
it immediately follows that
\[
A=B^\frac{1-\alpha}{2} C B^\frac{\alpha}{2}.
\]
But this implies $\rng{(A)}\subset \rng{(B)}$, a contradiction. Therefore, we have
\begin{equation*} \label{E:M2}
K_\alpha(A||B)= \left.
  \begin{cases}
    \tr B^{-\alpha}(A-B)B^{\alpha-1}(A-B), & \text{if } \supp{A} \subseteq \supp{B} \\
    \infty , & \text{otherwise.}
  \end{cases}
  \right.
\end{equation*}  

Since
\[
\tr B^{-\alpha}(A-B)B^{\alpha-1}(A-B)= \tr \left|B^\frac{\alpha -1}{2}(A-B)B^\frac{-\alpha}{2}\right|^2,
\]
it follows easily that $K_\alpha(A||B)\geq 0$ for any $A,B \in \mc{L}^+(\cH)$ and $K_\alpha(A||B)=0$ holds if and only if $A=B.$ Therefore, the $\chi_\alpha^2$-divergence is always nonnegative and take the value 0 only at identical operators. This means that $K_\alpha(.||.)$ is really a divergence or, in other words, a generalized distance measure.

We also note that in the special case where $\alpha \in \{0,1\}$, the $\chi_\alpha^2$-divergence coincides with the so-called \emph{quadratic relative entropy.}
The transformations of the state space $\cS(\cH)$ and those of the set $\cM\ler{\cH}$ of all nonsingular density operators leaving the quadratic relative entropy invariant have been determined in \cite{molnagy}, see Theorems 2 and 3.

In the main result of this paper we show that all bijective maps of the positive definite cone $\mc{L}^{++}(\cH)$ which preserve the $\chi_\alpha^2$-divergence for some $\alpha\in [0,1]$ are unitary or antiunitary similarity transformations.
\end{remark}

We remark that in \cite{Ruskaietal} an even more general concept of $\chi^2$-divergence depending on a function parameter was also defined in the manner of $f$-divergences, see equation (10) on page 122201-3. Important properties of these notions (both the more restricted one given in Definition~\ref{D:1} as well as the just mentioned more general one) were investigated in several further papers. Without presuming to be exhaustive here we refer only to the works \cite{Hansen}, \cite{HP}, \cite{Jencova}, \cite{PG}, \cite{Temmeetal}.  

Before presenting our results we would like to make clear a point.
In the light of our structural results given in \cite{mnsz13} and \cite{dv16b} concerning maps preserving $f$-divergences, one may immediately put the question that what about the preservers of the general notion of $\chi^2$-divergence. The honest answer is that we do not know. As the reader will see below, compared to the above listed previous results of ours, the description even in the considered case of $\chi^2_\alpha$-divergences is remarkably more difficult and more complicated requiring the invention of many new ideas. Presently, we do not see any ways how one can attack the general problem.

\section{The main results}

In this section we present the main results of the paper.
Select an arbitrary number $\alpha\in[0,1]$. It is clear that for any unitary or antiunitary operator $U$ on $\mc{H}$, the corresponding conjugation 
\[
A\mapsto UAU^* \qquad \ler{A\in  \mc{L}^{+}(\cH)}
\]
leaves the quantum $\chi_\alpha^2$-divergence invariant.
In our results we show that the converse statement is also true, i.e., the preservers of the quantum $\chi_\alpha^2$-divergence are all necessarily unitary or antiunitary conjugations.

The precise formulations of the statements read as follows. We begin with the case of the positive definite cone.

\begin{theorem} \label{fo}
Let $\alpha \in [0,1]$ be an arbitrary but fixed number and let $\phi: \mc{L}^{++}(\cH) \rightarrow \mc{L}^{++}(\cH)$ be a bijection which preserves the quantum $\chi_\alpha^2$-divergence, that is, satisfies
$$
K_\alpha(\phi(A)||\phi(B))=K_\alpha(A||B) \qquad \ler{A, B \in \mc{L}^{++}(\cH)}.
$$
Then there exists a unitary or an antiunitary operator $U: \cH \rightarrow \cH$ such that
$$
\phi(A)=UAU^* \qquad \ler{A \in \mc{L}^{++}(\cH)}.
$$
\end{theorem}

The theorem will be proven in a separate section.
We next formulate the corresponding results concerning the cone of positive semidefinit operators and the state space.

\begin{proposition} \label{semidef}
Let $\alpha \in [0,1]$ be an arbitrary but fixed number and let $\phi: \mc{L}^{+}(\cH) \rightarrow \mc{L}^{+}(\cH)$ be a bijection such that
$$
K_\alpha(\phi(A)||\phi(B))=K_\alpha(A||B) \qquad \ler{A, B \in \mc{L}^{+}(\cH)}.
$$
Then there is a unitary or an antiunitary operator $U: \cH \rightarrow \cH$ such that
$$
\phi(A)=UAU^* \qquad \ler{A \in \mc{L}^{+}(\cH)}.
$$
\end{proposition}

\begin{proposition} \label{density}
Let $\alpha \in [0,1]$ be an arbitrary but fixed number and let $\phi: \cS(\cH) \rightarrow \cS(\cH)$ be a bijection such that
$$
K_\alpha(\phi(A)||\phi(B))=K_\alpha(A||B) \qquad \ler{A, B \in \cS(\cH)}.
$$
Then there exists a unitary or an antiunitary operator $U: \cH \rightarrow \cH$ such that
$$
\phi(A)=UAU^* \qquad \ler{A \in \cS(\cH)}.
$$
\end{proposition}

The proofs of the latter two propositions can be obtained by using arguments similar to the ones that we will employ in the proof of Theorem \ref{fo}. Therefore, we will only sketch those proofs in the last part of the next section.

\section{Proofs}

This section is devoted to the proofs of our results.
However, let us begin with the following remark. We have already mentioned that in our previous works \cite{mpv15} and \cite{dv16b} we presented structural results concerning maps on the positive definite or semidefinite cones preserving Bregman divergences, or Jensen divergences, or $f$-divergences. Therefore, it is necessary to make it clear that what we obtain in the present paper, our main result Theorem \ref{fo} is a really new result, it is independent from the previous ones. So, we need to verify that the $\chi_\alpha^2$-divergences we consider here are neither $f$-divergences (with the exception of the cases $\alpha =0$, $\alpha =1$), nor Bregman or Jensen divergences on the set of all positive definite operators.

Indeed, these are very easy to see. As for $f$-divergences (see, e.g., Section 2 in \cite{HMPB}),  write $A=tI$, $t>0$ and $B=I$ into
\[
S_f(A||B)=K_\alpha (A||B)
\]
and obtain that $f(t)=(t-1)^2$, $t>0$. It then follows that
\[
S_f(A||B)=K_0(A||B)
\]
holds for all $A,B\in \mc{L}^{++}(\mc{H})$ which implies that $\alpha=0$ or $\alpha=1$.

As for Bregman divergences (see, e.g., Section 1 in \cite{mpv15}), we do something similar. 
We write $A=tI$, $t>0$ and $B=sI$, $s>0$ into
\[
H_f(A||B)=K_\alpha (A||B)
\]
and, for $s=1$, conclude that $f$ is a quadratic function. Letting now $s$ vary, we see that the left hand side of the equality above is quadratic in $s$, while the right hand side is not so. This gives a contradiction.

Finally, as for Jensen divergences, it is clear that they are symmetric in their variables while the $\chi^2_\alpha$-divergences are not so. Consequently, the results of the present paper are really new. In fact, as can be seen from the arguments to be given below, the proofs are more deep and involved than any of the previous results we have obtained so far in this line of research.

In the next pages we present the proof of Theorem \ref{fo}. For the sake of transparency, we divide it into three parts given in the following three subsections the first two parts being split into several substeps.

\subsection{Proof of Theorem \ref{fo} --- part one} \label{1fej}

In what follows, let $\alpha \in [0,1]$ be an arbitrary but fixed number and let $\phi: \mc{L}^{++}(\cH) \rightarrow \mc{L}^{++}(\cH)$ be a bijective map such that
$$
K_\alpha(\phi(A)||\phi(B))=K_\alpha(A||B) \qquad \ler{A, B \in \mc{L}^{++}(\cH)}.
$$

In the first part of the proof we show that $\phi$ is a homeomorphism and it can be extended to a map $\psi$ on the set $\mc{L}^{+}(\cH).$ (We make a remark here: observe that $\mc{L}^{sa}(\cH)$ is a finite dimensional linear space, hence there is only one locally convex Hausdorff vector topology on it, the topology of the operator norm, and whenever we use topological notions we always mean that unique topology.)
Furthermore, we also verify that the extension $\psi$ is bijective on $\mc{L}^+(\cH)$, it \emph{almost} preserves the $\chi_\alpha^2$-divergence (for the definition of this notion see Claim \ref{psziorz}), and it preserves the trace.

In what follows we will need the continuity properties of the $\chi^2_\alpha$-divergences what we collect below.

\begin{remark} \label{mjfoly}
First, it is clear that the map
$$
K_\alpha(.||B): \, \mc{L}^+(\cH)\rightarrow [0,\infty); \, A \mapsto K_\alpha(A||B)=\tr B^{-\alpha}(A-B)B^{\alpha-1}(A-B)
$$
is continuous for any fixed $B \in \mc{L}^{++}(\cH),$
and the map
$$
K_\alpha(A||.): \, \mc{L}^{++}(\cH)\rightarrow [0,\infty); \, B \mapsto K_\alpha(A||B)=\tr B^{-\alpha}(A-B)B^{\alpha-1}(A-B)
$$
is continuous for any fixed $A \in \mc{L}^{+}(\cH).$

We remark that $K_\alpha(.||.)$ is not continuous in its first variable when the second variable is a singular element of 
$\mc{L}^+(\cH)$. To see this simple statement, let $B=P$ be a rank-one projection and set $A_n=P+(1/n)I$ $(n\in \mathbb N)$, a sequence which converges to $A=P$. Then we have $K_\alpha(A_n||B)=\infty$ for all $n\in \mathbb N$ but $K_\alpha (A||B)=0$.

We next show that the $\chi_\alpha^2$-divergence is not continuous on $\mc{L}^+(\cH)$ in its second variable. We consider only the case where $\alpha=0$ or $\alpha=1$, that is when $K_\alpha(.||.)$ is an $f$-divergence. In the remaining cases one can argue in a similar way.

In fact, we have the discontinuity already in two-dimension. To this, set
\[
B_n^{1/2}=
\left[
\begin{matrix}
1 & \frac{1}{n} \\ \frac{1}{n} & \frac{2}{n^2}
\end{matrix}
\right]
\]  
for every $n\in \mathbb N$ and
\[
P=
\left[
\begin{matrix}
1 & 0 \\ 0 & 0
\end{matrix}
\right].
\]  
We clearly have $B_n\to P$ and one can verify
\[
K_0 (P|| B_n)=4+n^2 -2 + \ler{1+\frac{2}{n^2} +\frac{4}{n^4}} \to \infty
\]
although $K_\alpha(P||P)=0$.

We mention that this example shows that the statement Proposition 2.12 in \cite{HMPB} asserting that the $f$-divergences are continuous on $\mc{L}^+(\cH)$ in their second variables is false.
\end{remark}

We use the above mentioned continuity properties of $K_\alpha(.||.)$ to prove the following statement.

\begin{claim} \label{folyt}
The map $\phi$ is a homeomorphism.
\end{claim}
\begin{proof}
Since $\phi$ is a bijection on $\mc{L}^{++}(\cH)$ preserving the $\chi_\alpha^2$-divergence, so is its inverse $\phi^{-1}.$ Therefore, we need only to show that $\phi$ is continuous. Let $\{A_n\}_{n \in \N}$ be a sequence in $\mc{L}^{++}(\cH)$ which converges to some $A \in \mc{L}^{++}(\cH).$ By the continuity of the map $K_\alpha(.||A)$, see Remark \ref{mjfoly}, we have
$$
\lim_{n \to \infty} K_\alpha\ler{A_n||A}=K_\alpha \ler{\lim_{n \to \infty} A_n\middle|\middle|A}=K_\alpha(A||A)=0.
$$
Therefore,
$$\lim_{n \to \infty} K_\alpha\ler{\phi\ler{A_n}||\phi\ler{A}}=0.$$
On the other hand, we compute
\begin{equation}\label{homeo}
\begin{gathered}
K_\alpha\ler{\phi\ler{A_n}||\phi\ler{A}}
\\
=\tr \ler{\phi(A)^{-\frac{\alpha}{2}}\ler{\phi\ler{A_n}-\phi(A)}\phi(A)^{\frac{\alpha-1}{2}}}
\\
\cdot \ler{\phi(A)^{-\frac{\alpha}{2}}
\ler{\phi\ler{A_n}-\phi(A)}\phi(A)^{\frac{\alpha-1}{2}}}^*
\\
\geq \norm{\phi(A)^{-\frac{\alpha}{2}}\ler{\phi\ler{A_n}-\phi(A)}\phi(A)^{\frac{\alpha-1}{2}}}_{op}^2
\\ 
\geq \frac{\norm{\phi\ler{A_n}-\phi(A)}_{op}^2}{\norm{\phi(A)^{\frac{\alpha}{2}}}_{op}^2 \norm{\phi(A)^{\frac{1-\alpha}{2}}}_{op}^2}
=\frac{\norm{\phi\ler{A_n}-\phi(A)}_{op}^2}{\norm{\phi(A)}_{op}}.
\end{gathered}
\end{equation}
The first inequality holds because the Hilbert-Schmidt norm majorizes the operator norm, and the second inequality holds because of the submultiplicativity of the operator norm. The term $\norm{\phi(A)}_{op}$ is independent of $n,$ hence we conclude that $\phi\ler{A_n} \to \phi(A)$ proving the claim.
\end{proof}

The following assertion is a sort of identification lemma relative to the set $\mc{M}(\mc{H})$ of all nonsingular states.

\begin{lemma}\label{L:M}
Assume $A,B\in \mc{L}(\mc{H})^{+}$ are such that for all $C\in \mc{M}(\mc{H})$ we have
\begin{equation}\label{E:M4}
K_\alpha(A||C)=K_\alpha(B||C).
\end{equation}
Then we obtain $A=B$. 
\end{lemma}

\begin{proof}
By \eqref{E:M4} we have
\[
\tr \ler{C^{-\alpha}AC^{\alpha -1} A}-2\tr A=
\tr \ler{C^{-\alpha}BC^{\alpha -1} B}-2\tr B
\]
for all $C\in \mc{M}(\mc{H})$. Pick arbitrary rank one projection $P\in \mc{P}_1(\mc{H})$ and let $Q=I-P$. For any $t\in (0,1)$ insert $tP+\frac{1-t}{d-1}Q$ into the place of $C$ in the displayed formula above ($d$ is the dimension of $\mc{H}$).

First consider the case where $\alpha=0$. We have
\begin{equation*}
\begin{aligned}
\frac{1}{t}\tr APA +\frac{d-1}{1-t}\tr AQA -2\tr A=
\frac{1}{t}\tr BPB +\frac{d-1}{1-t}\tr BQB -2\tr B
\end{aligned}
\end{equation*}
Since the functions $1, \frac{1}{t}, \frac{1}{1-t}$ are linearly independent over the interval $(0,1)$, it follows that
\[
\tr A^2 P=\tr APA=\tr BPB=\tr B^2 P
\]
holds for every rank-one projection $P$ on $\mc{H}$ which implies $A^2=B^2$ and then we deduce $A=B$. The same reasoning works for the case where $\alpha=1$. 

Now, let $\alpha \in (0,1)$. Again, the argument is practically the same but the computation is a bit more complicated. For $P,Q$ given as above and for any $t\in (0,1)$
we have
\begin{equation*}
\begin{gathered}
\frac{1}{t}\tr PAPA +
t^{-\alpha}\ler{\frac{1-t}{d-1}}^{\alpha-1}\tr PAQA\\+
\ler{\frac{1-t}{d-1}}^{-\alpha}t^{\alpha-1}\tr QAPA+
\frac{d-1}{1-t}\tr QAQA -2\tr A\\=
\frac{1}{t}\tr PBPB +
t^{-\alpha}\ler{\frac{1-t}{d-1}}^{\alpha-1}\tr PBQB\\+
\ler{\frac{1-t}{d-1}}^{-\alpha}t^{\alpha-1}\tr QBPB+
\frac{d-1}{1-t}\tr QBQB -2\tr B.
\end{gathered}
\end{equation*}
Using the linear independence of the functions
\[
1, \frac{1}{t}, \frac{1}{1-t}, t^{-\alpha}\ler{1-t}^{\alpha-1}, \ler{1-t}^{-\alpha}t^{\alpha-1}
\]
over the interval $(0,1)$
in the case where $\alpha\neq 1/2$ (if $\alpha=1/2$, the last two functions are the same) we get that
$\tr PAPA= \tr PBPB$ holds for every  rank-one projection $P$ on $\mc{H}$ which easily gives us that $A=B$. The case $\alpha=1/2$ can be treated in the same way.
\end{proof}

\begin{claim} \label{kiterj}
Let $\{A_n\}_{n \in \N}$ be a convergent sequence of positive definite operators on $\mc{H}$ and let us denote its limit by $A.$ (Clearly, $A \in \mc{L}^{+}(\cH).$) Then $\{\phi\ler{A_n}\}_{n\in \mathbb N}$ is convergent and, obviously, $\lim_{n \to \infty} \phi\ler{A_n} \in \mc{L}^+(\cH).$
Consequently, it follows that if $\{A_n\}_{n \in \N}$ and $\{B_n\}_{n \in \N}$ are convergent sequences of positive definite operators on $\mc{H}$ such that $\lim_{n \to \infty} A_n=\lim_{n \to \infty} B_n,$ then we have $\lim_{n \to \infty} \phi\ler{A_n}=\lim_{n \to \infty} \phi\ler{B_n}.$
\end{claim}

\begin{proof} Let $\{A_n\}_{n \in \N}$ be a convergent sequence of positive definite operators on  $\mc{H}$ with $\lim_{n \to \infty} A_n=A \in \mc{L}^+(\cH)$ and let $X$ be an arbitrary element of $\mc{L}^{++}(\cH).$ Then, by Remark \ref{mjfoly},
\be \label{konv1}
K_\alpha\ler{A||X}=\lim_{n \to \infty} K_\alpha\ler{A_n||X}=\lim_{n \to \infty} K_\alpha\ler{\phi\ler{A_n}||\phi(X)}.
\ee
By the inequality \eqref{homeo} in Claim \ref{folyt} we have
$$
K_\alpha\ler{\phi\ler{A_n}||\phi(X)} \geq \frac{\norm{\phi\ler{A_n}-\phi\ler{X}}_{op}^2}{\norm{\phi(X)}_{op}}.
$$
The left hand side of the above inequality is convergent and hence bounded, so $\norm{\phi\ler{A_n}-\phi\ler{X}}_{op}$ is bounded as well. Therefore, the sequence $\{\phi\ler{A_n}\}_{n \in \N}$ is bounded. Assume that $\{\phi\ler{A_n}\}_{n \in \N}$ has two accumulation points, say $B_1$ and $B_2.$ That is, we have
$$
\lim_{n\to \infty} \phi(A_{k_n})=B_1
\text{ and }
\lim_{n\to \infty} \phi(A_{l_n})=B_2
$$
for some subsequences
$\{\phi(A_{k_n})\}_{n \in \N}$ and $\{\phi(A_{l_n})\}_{n \in \N}$.
The $\chi_\alpha^2$-diver\-gence is continuous in its first variable when the second variable is nonsingular, Remark \ref{mjfoly}, hence
$$
K_\alpha\ler{B_1||\phi(X)}=\lim_{n \to \infty} K_\alpha\ler{\phi(A_{k_n})||\phi(X)}
$$
and
$$
K_\alpha\ler{B_2||\phi(X)}=\lim_{n \to \infty} K_\alpha\ler{\phi(A_{l_n})||\phi(X)}
$$
hold for any $X \in \mc{L}^{++}(\cH).$
However, the right hand sides of the above equations coincide as the sequence $\{K_\alpha(\phi\ler{A_n}||\phi(X))\}_{n\in \mathbb N}$ is convergent, see \eqref{konv1}. So, we deduced that
\begin{equation*} \label{egyen}
K_\alpha\ler{B_1||\phi(X)}=K_\alpha\ler{B_2||\phi(X)}
\end{equation*}
for any $X \in \mc{L}^{++}(\cH).$ 
By Lemma~\ref{L:M} we obtain that $B_1=B_2$.
It follows that the sequence $\{\phi(A_n)\}_{n\in \mathbb N}$ is convergent.
\par
We can easily show the rest of the statement, that is, that $\lim_{n \to \infty} A_n=\lim_{n \to \infty} B_n$ implies $\lim_{n \to \infty} \phi\ler{A_n}=\lim_{n \to \infty} \phi\ler{B_n}.$ Indeed, assume that $\{A_n\}_{n \in \N}$ and $\{B_n\}_{n \in \N}$ are convergent sequences of positive definite operators such that $\lim_{n \to \infty} A_n=\lim_{n \to \infty} B_n.$

Let the sequence $\{C_n\}_{n \in \N}$ be defined by
$$
C_{2n}:=A_n \text{ and } C_{2n+1}:=B_n \qquad \ler{n \in \N}.
$$
Clearly, $\{C_n\}_{n\in \mathbb N}$ is a convergent sequence of positive definite operators, hence by the first part of this Claim (which has been already proven) the sequence $\{\phi\ler{C_n}\}_{n\in \mathbb N}$ is also convergent. Therefore, any subsequence of $\{\phi\ler{C_n}\}_{n \in \N}$ is convergent and has the same limit. In particular, we have $\lim_{n \to \infty} \phi\ler{A_n}=\lim_{n \to \infty} \phi\ler{B_n}.$
\end{proof}

\begin{remark} \label{inverz}
Observe that the statements in Claim \ref{kiterj} hold also for $\phi^{-1}$ as the latter map is also a $\chi_\alpha^2$-divergence preserving bijection just like $\phi.$
\end{remark}

Now we are in the position to define the map
\be \label{pszidef}
\psi: \mc{L}^+(\cH) \rightarrow \mc{L}^+(\cH); \, A \mapsto \psi(A):=\lim_{Z \to A, Z \in \mc{L}^{++}(\cH)} \phi(Z).
\ee

The definition in \eqref{pszidef} is correct by Claim \ref{kiterj}. The map $\phi$ is continuous by Claim \ref{folyt}, hence $\psi(B)=\phi(B)$ for any $B\in \mc{L}^{++}(\cH).$ It follows that $\psi$ is an extension of $\phi$.
\par
By Remark \ref{inverz}, we can also define the transformation
\begin{equation*} \label{pszicsilldef}
\psi^*: \mc{L}^+(\cH) \rightarrow \mc{L}^+(\cH); \, A \mapsto \psi^{*}(A):=\lim_{Z \to A, Z \in \mc{L}^{++}(\cH)} \phi^{-1}(Z).
\end{equation*}

\begin{claim} \label{psziinverz}
The above defined map $\psi^*$ is the inverse of $\psi,$ that is, $\psi \circ \psi^*=\psi^*\circ \psi=\mathrm{id}_{\mc{L}^+(\cH)}.$
In particular, $\psi$ is bijective.
\end{claim}
\begin{proof}
We only show that $\psi^*\circ \psi=\mathrm{id}_{\mc{L}^+(\cH)}$ as the equality $\psi \circ \psi^*=\mathrm{id}_{\mc{L}^+(\cH)}$ can be proven very similarly.
Let $A \in \mc{L}^+(\cH)$ be arbitrary and let $\{A_n\}_{n \in \N}$ be a sequence of positive definite operators on $\mc{H}$ with $\lim_{n \to \infty} A_n=A.$
Then
$$
\psi(A)=\lim_{n \to \infty} \phi\ler{A_n}
$$
and thus, by the definition of $\psi^*$, we have
$$
\psi^*\ler{\psi(A)}=\psi^*\ler{\lim_{n \to \infty} \phi\ler{A_n}}=\lim_{n \to \infty} \phi^{-1}\ler{\phi\ler{A_n}}=A.
$$
\end{proof}

\begin{claim} \label{psziorz}
The transformation $\psi$ \emph{almost} preserves the $\chi_\alpha^2$-divergence by what we mean that
\begin{equation*} \label{psziorzegy}
K_\alpha\ler{\psi(A)||\psi(B)}= K_\alpha\ler{A||B}
\end{equation*}
holds for any $A \in \mc{L}^+(\cH)$ and $B \in \mc{L}^{++}(\cH).$
The same is true for the map $\psi^{-1}.$
\end{claim}
\begin{proof}
Pick $A \in \mc{L}^+(\cH), B \in \mc{L}^{++}(\cH)$ and let $\{A_n\}_{n \in \N}$ be an arbitrary sequence of positive definite operators on $\mc{H}$ converging to $A.$ Then
$$
K_\alpha \ler{\psi(A)||\psi(B)}=K_\alpha \ler{\psi(A)||\phi(B)}=K_\alpha\ler{\lim_{n \to \infty} \phi\ler{A_n}||\phi(B)}
$$
$$
=\lim_{n \to \infty} K_\alpha\ler{\phi\ler{A_n}||\phi(B)}=\lim_{n \to \infty} K_\alpha\ler{A_n||B}=K_\alpha\ler{A||B}.
$$
The verification of
$$
K_\alpha \ler{\psi^{-1}(A)||\psi^{-1}(B)}=K_\alpha\ler{A||B} \qquad \ler{A \in \mc{L}^+(\cH), B \in \mc{L}^{++}(\cH)}
$$
is similar.
\end{proof}

\begin{claim} \label{phitrace}
The map $\phi: \mc{L}^{++}(\cH) \rightarrow \mc{L}^{++}(\cH)$ preserves the trace, that is,
$$
\tr \phi(C)=\tr C \qquad \ler{C \in \mc{L}^{++}(\cH)}.
$$
\end{claim}

\begin{proof}
For any positive definite operators $B, C$ and $X$ on $\mc{H}$ we have
$$
K_\alpha \ler{X||B}-K_\alpha \ler{X||C}=\tr B^{-\alpha} X B^{\alpha-1}X-\tr C^{-\alpha} X C^{\alpha-1}X+\tr B-\tr C.
$$
Therefore, the set
$$
\left\{K_\alpha \ler{X||B}-K_\alpha \ler{X||C} \middle| X \in \mc{L}^{++}(\cH)\right\}
$$
is bounded from below if and only if the inequality
$$
\tr B^{-\alpha} X B^{\alpha-1}X \geq \tr C^{-\alpha} X C^{\alpha-1}X
$$
holds for every positive definite $X.$ Moreover, in this case we clearly have
$$
\mathrm{inf}\left\{K_\alpha \ler{X||B}-K_\alpha \ler{X||C} \middle| X \in \mc{L}^{++}(\cH)\right\}=\tr B-\tr C.
$$
Now, assume that $B \leq C$ holds for some positive definite operators $B$ and $C$ on $\mc{H}$, that is, $C-B$ is positive semidefinite. Then
$$\left\{K_\alpha \ler{X||B}-K_\alpha \ler{X||C} \middle| X \in \mc{L}^{++}(\cH)\right\}$$
is bounded from below. Indeed, by the L\"owner-Heinz theorem (see, e.g., \cite[Theorem 2.6]{carlen}), the map $t \mapsto t^p$ is operator monotone decreasing on $(0,\infty)$ for any $p \in [-1, 0].$ Therefore, $B^{-\alpha} \geq C^{-\alpha}$ and $B^{\alpha-1} \geq C^{\alpha-1}$ for any $\alpha \in [0,1].$ Consequently, $X^{\frac{1}{2}}B^{-\alpha} X^{\frac{1}{2}} \geq X^{\frac{1}{2}} C^{-\alpha} X^{\frac{1}{2}}$ and $X^{\frac{1}{2}} B^{\alpha-1} X^{\frac{1}{2}} \geq X^{\frac{1}{2}}C^{\alpha-1}X^{\frac{1}{2}}$ holds for any $X \in \mc{L}^{++}(\cH).$ It is folklore that an operator $A \in \mc{L}^{sa}(\cH)$ is positive semidefinite if and only if $\tr A T\geq 0$ for any $T \in \mc{L}^+(\cH).$ Therefore,
$$
\tr B^{-\alpha} X B^{\alpha-1}X=\tr \ler{X^{\frac{1}{2}}B^{-\alpha} X^{\frac{1}{2}}}\ler{X^{\frac{1}{2}}B^{\alpha-1} X^{\frac{1}{2}}}
$$
$$
\geq \tr \ler{X^{\frac{1}{2}}C^{-\alpha} X^{\frac{1}{2}}}\ler{X^{\frac{1}{2}}B^{\alpha-1} X^{\frac{1}{2}}}
\geq \tr \ler{X^{\frac{1}{2}}C^{-\alpha} X^{\frac{1}{2}}}\ler{X^{\frac{1}{2}}C^{\alpha-1} X^{\frac{1}{2}}}
$$
$$
=\tr C^{-\alpha} X C^{\alpha-1}X
$$
holds, so we have the required boundedness from below.
\par
Let us now pick some $C, D \in \mc{L}^{++}(\cH)$ and choose an $\eps>0$ such that $\eps I\leq C$ and $\eps I\leq D.$
Then, as we have seen above,
$$
\left\{K_\alpha \ler{X||\eps I}-K_\alpha \ler{X||C} \middle| X \in \mc{L}^{++}(\cH)\right\}
$$
is bounded from below and
$$
\mathrm{inf}\left\{K_\alpha \ler{X||\eps I}-K_\alpha \ler{X||C} \middle| X \in \mc{L}^{++}(\cH)\right\}=\tr \eps I-\tr C.
$$
Similarly,
$$
\left\{K_\alpha \ler{X||\eps I}-K_\alpha \ler{X||D} \middle| X \in \mc{L}^{++}(\cH)\right\}
$$
is bounded from below and
$$
\mathrm{inf}\left\{K_\alpha \ler{X||\eps I}-K_\alpha \ler{X||D} \middle| X \in \mc{L}^{++}(\cH)\right\}=\tr \eps I-\tr D.
$$
By the bijectivity of $\phi$ one can see that
$$
\left\{K_\alpha \ler{X||\eps I}-K_\alpha \ler{X||C} \middle| X \in \mc{L}^{++}(\cH)\right\}
$$
$$
=\left\{K_\alpha \ler{X||\phi\ler{\eps I}}-K_\alpha \ler{X||\phi\ler{C}} \middle| X \in \mc{L}^{++}(\cH)\right\},
$$
hence the latter set is also bounded from below, and its infimum is $\tr \eps I-\tr C.$ On the other hand, by the first observation of the present proof, this infimum is equal also to $\tr \phi\ler{\eps I}-\tr \phi(C).$ 
Hence we have
$$
\tr \eps I-\tr C=\tr \phi\ler{\eps I}-\tr \phi(C)
$$
and very similarly we get
$$
\tr \eps I-\tr D=\tr \phi\ler{\eps I}-\tr \phi(D).
$$
Consequently,
$$
\tr \phi(C)-\tr C=\tr \phi(D)-\tr D.
$$
The operators $C$ and $D$ were arbitrary, so we derive that
$$
\tr \phi(C)=\tr C+\delta \qquad \ler{C \in \mc{L}^{++}(\cH)}
$$
for some $\delta \in \R$ which is independent of $C.$ Clearly, $\delta<0$ is impossible and the bijectivity of $\phi$ excludes the possibility $\delta>0.$ So we infer that $\delta =0$ implying $\tr \phi(C)=\tr(C)$ for any $C \in \mc{L}^{++}(\cH).$
\end{proof}

\begin{claim} \label{psitrace}
The map $\psi: \mc{L}^{+}(\cH) \rightarrow \mc{L}^{+}(\cH)$ also preserves the trace, that is, we have
$$
\tr \psi(A)=\tr A \qquad \ler{A \in \mc{L}^{+}(\cH)}.
$$
\end{claim}
\begin{proof}
Let $A \in \mc{L}^{+}(\cH)$ and select a sequence
$\{A_n\}_{n \in \N}$ of positive definite operators on $\cH$ converging to $A$.
Then by the trace preserving property of the map $\phi$, Claim \ref{phitrace}, and by the continuity of the trace functional, we have
$$
\tr \psi(A)=\tr \lim_{n \to \infty} \phi\ler{A_n}=\lim_{n \to \infty} \tr \phi\ler{A_n}=\lim_{n \to \infty} \tr A_n =\tr A.
$$
\end{proof}
\begin{remark}
Clearly, the transformation $\psi^{-1}$ also preserves the trace.
\end{remark}

\subsection{Proof of Theorem \ref{fo} --- part two} \label{2fej}

In this subsection let $\xi: \cS(\cH)\rightarrow \cS(\cH)$ be a bijective map such that $\xi(\cM(\cH))=\cM(\cH)$ and assume that
\begin{equation}\label{E:M3}
K_\alpha(\xi (A)|| \xi(B))=K_\alpha (A||B)
\end{equation}
holds for $A\in \mc{S}(\mc{H})$ and $B\in \mc{M}(\mc{H})$. 

In what follows we prove that $\xi$ equals a unitary or an antiunitary conjugation on $\mc{M}(\mc{H})$.

\begin{claim} \label{szigkonv}
Let $B \in \mc{L}^{++}(\cH)$ be fixed. The map
$$
K_\alpha(.||B): \, \mc{L}^+(\cH)\rightarrow [0,\infty); \, A \mapsto K_\alpha(A||B)=\tr B^{-\alpha}(A-B)B^{\alpha-1}(A-B)
$$
is strictly convex.
\end{claim}
\begin{proof} Indeed, we have
$$
K_\alpha(A||B)=\tr B^{-\alpha} A B^{\alpha-1}A-2 \tr A+\tr B
=\norm{B^{-\frac{\alpha}{2}}A B^{\frac{\alpha-1}{2}}}_{HS}^2-\ler{2 \tr A-\tr B}.
$$
The first term is strictly convex in $A$ since the Hilbert-Schmidt norm is strictly convex, and the second term is affine in $A$. This implies the assertion.
\end{proof}

\begin{claim} \label{pszibij}
The map $\xi$ restricted to $\cP_1(\cH)$ is a bijection from $\cP_1(\cH)$ onto itself.
\end{claim}
\begin{proof}
We have seen in Remark \ref{mjfoly} and Claim \ref{szigkonv} that the map $K_\alpha(.||B): \, \mc{L}^+(\cH)\rightarrow [0,\infty); \, A \mapsto K_\alpha(A||B)$ is continuous and strictly convex on $\mc{L}^+(\cH)$ for any fixed $B \in \mc{L}^{++}(\cH).$ Therefore, so is the restriction of $K_\alpha(.||B)$ to the compact and convex set $\cS(\cH) \subset \mc{L}^+(\cH)$ of states.
\par
Recall that $d$ denotes the dimension of the Hilbert space $\cH.$
On the one hand, if
$$
K_\alpha\ler{P \middle|\middle| \frac{1}{d}I}= \mathrm{max}\left\{K_\alpha\ler{X \middle|\middle| \frac{1}{d}I}\middle|X \in \cS(\cH)\right\}
$$
holds for some $P \in \cS(\cH)$, then $P$ is an extremal point of $\cS(\cH)$ by the strict convexity of $K_\alpha\ler{.\middle|\middle|\frac{1}{d}I}.$ This implies that $P \in \cP_1(\cH).$ On the other hand, for any $P,Q \in \cP_1(\cH)$ there exists some unitary $U \in \mc{L}(\cH)$ such that $Q=UPU^*.$ By the clear unitary invariance of the $\chi_\alpha^2$-divergence, we have
$$
K_\alpha\ler{Q \middle|\middle| \frac{1}{d}I}=K_\alpha\ler{UPU^* \middle|\middle| U\frac{1}{d}I U^*}=K_\alpha\ler{P \middle|\middle| \frac{1}{d}I}.
$$
Therefore, $K_\alpha\ler{.\middle|\middle|\frac{1}{d}I}$ is constant on $\cP_1(\cH)$ which means that for any $P \in \cP_1(\cH)$ we have
\be \label{p1kar}
K_\alpha\ler{P \middle|\middle| \frac{1}{d}I}= \mathrm{max}\left\{K_\alpha\ler{X \middle|\middle| \frac{1}{d}I}\middle|X \in \cS(\cH)\right\}.
\ee
We deduce that for $P \in \cS(\cH)$ we have $P \in \cP_1(\cH)$ if and only if \eqref{p1kar} holds.
\par
By the preserver property of $\xi$ we have
$$
K_\alpha\ler{X \middle|\middle| \frac{1}{d}I}=K_\alpha\ler{\xi(X) \middle|\middle| \xi\ler{\frac{1}{d}I}}
$$
for any $X \in \cS(\cH).$ Therefore,
$$
P \in \cP_1(\cH) \Rightarrow K_\alpha\ler{P \middle|\middle| \frac{1}{d}I}= \mathrm{max}\left\{K_\alpha\ler{X \middle|\middle| \frac{1}{d}I}\middle|X \in \cS(\cH)\right\}
$$
$$
\Rightarrow K_\alpha\ler{\xi(P) \middle|\middle| \xi \ler{\frac{1}{d}I}}= \mathrm{max}\left\{K_\alpha\ler{X \middle|\middle| \xi \ler{\frac{1}{d}I}}\middle|X \in \cS(\cH)\right\}
$$
$$
\Rightarrow \xi(P) \in \cP_1(\cH),
$$
because the map $K_\alpha(.||\xi(\frac{1}{d}I))$ is also strictly convex by Claim \ref{szigkonv}.
\par
Consequently, we obtain that $\xi\ler{\cP_1(\cH)} \subseteq \cP_1(\cH).$ In the above argument we can replace $\xi$ by $\xi^{-1},$ hence $\xi^{-1}\ler{\cP_1(\cH)} \subseteq \cP_1(\cH)$ also holds. This means that $\xi$ maps $\cP_1(\cH)$ bijectively onto itself.
\end{proof}

We note that the computation rule
\be \label{szamszab}
\tr RXRY=\tr RX \cdot \tr RY
\ee
can be verified by easy computation for any operators $X, Y \in \mc{L}(\cH)$ and for any rank one projection $R \in \cP_1(\cH).$ This will be used in the sequel several times.
\par
For the sake of simplicity, let us introduce the notation
$$
K_\alpha^*(X||Y):=K_\alpha(X||Y)+1 \qquad \ler{X\in \mc{L}^+(\cH), Y \in  \mc{L}^{++}(\cH)}.
$$
An easy but useful consequence of \eqref{szamszab} is that for any nonsingular density operator $D$ with spectral resolution
\be \label{ddef}
\cM(\cH) \ni D=\sum_{j=1}^m\lambda_j P_j, \enskip \lambda_1>\lambda_2>\dots>\lambda_m>0, \enskip \sum_{j=1}^m   \lambda_j \mathrm{rank}  \ler{P_j} =1
\ee
we have for every rank-one projection $R$ on $\mc{H}$ that
\begin{equation}\label{szamszab2}
\begin{gathered}
K_\alpha^*(R||D)=\tr R D^{-\alpha} R D^{\alpha-1}=\tr R D^{-\alpha} \cdot \tr R D^{\alpha-1}
\\
=\ler{\sum_{j=1}^m (\tr R P_j) \lambda_j^{-\alpha}} \ler{\sum_{k=1}^m (\tr R P_k) \lambda_k^{\alpha-1}}.
\end{gathered}
\end{equation}
The formula \eqref{szamszab2} clearly shows that
\be \label{minertek}
\mathrm{min}\left\{K_\alpha^*\ler{X || D} \middle| X \in \cP_1(\cH)\right\}=\frac{1}{\lambda_1}
\ee
and
\be \label{maxertek}
\mathrm{max}\left\{K_\alpha^*\ler{X || D} \middle| X \in \cP_1(\cH)\right\}=\frac{1}{\lambda_m}.
\ee
 Moreover, we have
\be \label{minelv}
K_\alpha^*(R||D)=\mathrm{min}\left\{K_\alpha^*\ler{X || D} \middle| X \in \cP_1(\cH)\right\}
\text{ if and only if }R \leq P_1,
\ee
and
\be \label{maxelv}
K_\alpha^*(R||D)=\mathrm{max}\left\{K_\alpha^*\ler{X || D} \middle| X \in \cP_1(\cH)\right\}
\text{ if and only if }R \leq P_m.
\ee

\begin{claim} \label{ortorz}
The map $\xi_{|\cP_1(\cH)}: \cP_1(\cH) \rightarrow \cP_1(\cH)$ preserves orthogonality in both directions.
\end{claim}

\begin{proof} 
Clearly, since $\xi$ and $\xi^{-1}$ have similar properties, it is enough to prove that $\xi$ preserves orthogonality only in one direction, i.e., it maps orthogonal rank-one projections to orthogonal ones.

Select $P, Q \in \cP_1(\cH)$ such that $PQ=0,$ and let
\be \label{bdef}
B:=\lambda P + \nu (I-(P+Q))+\mu Q,
\ee
where $1>\lambda>\nu>\mu>0,$ and $\lambda +(d-2) \nu+\mu=1.$ (Recall that $d$ denotes the dimension of the Hilbert space $\cH.$)
The operator $B$ defined in \eqref{bdef} is a nonsingular element of $\cS(\cH),$ hence by \eqref{minertek} and \eqref{minelv}, we have
$$
\mathrm{min}\left\{K_\alpha^*\ler{X || B} \middle| X \in \cP_1(\cH)\right\}=\frac{1}{\lambda},
$$
and the minimum is taken if and only if $X=P,$
and by \eqref{maxertek} and \eqref{maxelv} we also know that
$$
\mathrm{max}\left\{K_\alpha^*\ler{X || B} \middle| X \in \cP_1(\cH)\right\}=\frac{1}{\mu},
$$
and the maximum is taken if and only if $X=Q.$
\par
The transformation $\xi$ maps the set $\cP_1(\cH)$ bijectively onto itself, see Claim \ref{pszibij}, and satisfies \eqref{E:M3}. Hence
$$
\mathrm{min}\left\{K_\alpha^*\ler{X || \xi(B)} \middle| X \in \cP_1(\cH)\right\}=\frac{1}{\lambda}=K_\alpha^*\ler{\xi(P) || \xi(B)}
$$
and
$$
\mathrm{max}\left\{K_\alpha^*\ler{X || \xi(B)} \middle| X \in \cP_1(\cH)\right\}=\frac{1}{\mu}=K_\alpha^*\ler{\xi(Q) || \xi(B)}.
$$
Clearly, $\xi(B)$ is an invertible density operator. By \eqref{minelv} and \eqref{maxelv}, $K_\alpha^*\ler{R || \xi(B)}$ is minimal if and only if $R \leq P_\gamma,$ where $P_\gamma$ stands for the eigenprojection of $\xi(B)$ corresponding to the greatest eigenvalue, and $K_\alpha^*\ler{R || \xi(B)}$ is maximal if and only if $R \leq P_\sigma,$ where $P_\sigma$ stands for the eigenprojection of $\xi(B)$ corresponding to the smallest eigenvalue. (Observe that the greatest and the smallest eigenvalues of $\xi(B)$ can not coincide since $K_\alpha^*\ler{X || \xi(B)}$ takes the different values $1/\lambda , 1/\mu$ as $X$ runs through the set of rank-one projections.)
It follows that $\xi(P)$ and $\xi(Q)$ are subprojections of two different eigenprojections of $\xi(B)$ and hence we have $\xi(P)\xi(Q)=0.$ 
\end{proof}

\begin{claim} \label{trans_prob}
The map $\xi_{|\cP_1(\cH)}: \cP_1(\cH) \rightarrow \cP_1(\cH)$ preserves the transition probabilities meaning that it satisfies
$$
\tr \xi(P) \xi(R)= \tr PR \qquad \ler{P, R \in \cP_1(\cH)}.
$$
\end{claim}

\begin{proof}
Let $P\in \cP_1(\cH)$ and set
\begin{equation*} \label{cdef}
C:=\lambda P + \mu (I-P),
\end{equation*}
where $1>\lambda>\mu>0,$ and $\lambda +(d-1)\mu=1.$
Let $R \in \cP_1(\cH).$ Then
$$
R=P \Longleftrightarrow K_\alpha^*\ler{R||C}=\mathrm{min}\left\{K_\alpha^*\ler{X || C} \middle| X \in \cP_1(\cH)\right\}
$$
$$
\Longleftrightarrow K_\alpha^*\ler{\xi(R)||\xi(C)}=\mathrm{min}\left\{K_\alpha^*\ler{X || \xi(C)} \middle| X \in \cP_1(\cH)\right\}.
$$
This means that $\xi(P)$ is the one and only rank-one projection which is majorized by the eigenprojection of $\xi(C)$ corresponding to the greatest eigenvalue.
\par
We can easily see by \eqref{maxertek}, \eqref{maxelv} that $K_\alpha^*(R||C)=\frac{1}{\mu}$ for any $R \in \cP_1(\cH)$ which is orthogonal to $P.$ On the other hand, by Claim \ref{ortorz}, we have $PR=0 \Longleftrightarrow \xi(P)\xi(R)=0,$ so $K_\alpha^*\ler{X||\xi(C)}=\frac{1}{\mu}$ holds for any rank-one projection $X$ which is orthogonal to $\xi(P).$ This means by \eqref{minertek}-\eqref{maxelv} that
$$
\xi(C)= \lambda \xi(P) +\mu \ler{I-\xi(P)}.
$$
By \eqref{szamszab2} we get that
\begin{equation*}\label{szam_spec_2}
\begin{gathered}
K_\alpha^*\ler{\xi(R) || \xi(C)}
\\
=\ler{\tr \xi(R)\xi(P) \lambda^{-\alpha} +(1-\tr \xi(R) \xi(P))\mu^{-\alpha}}\times
\\ \times \ler{\tr \xi(R)\xi(P)\lambda^{\alpha-1}  +(1-\tr \xi(R) \xi(P))\mu^{\alpha-1}}
\end{gathered}
\end{equation*}
for any $R \in \cP_1(\cH).$
Comparing this to another consequence of \eqref{szamszab2}, namely, 
\begin{equation*}\label{szam_spec_1}
\begin{gathered}
K_\alpha^*\ler{R || C}\\=\ler{\tr R P \lambda^{-\alpha} +(1-\tr R P)\mu^{-\alpha}} \ler{\tr R P \lambda^{\alpha-1}  +(1-\tr R P)\mu^{\alpha-1}},
\end{gathered}
\end{equation*}
from the equality $K_\alpha^*\ler{\xi(R) || \xi(C)}=K_\alpha^*\ler{R || C}$ we can deduce that 
$$\tr \xi(P) \xi(R)= \tr PR$$ 
holds for any $R \in \cP_1(\cH).$ Indeed, to see this, it is enough to check that $K_\alpha^*\ler{\xi(R) || \xi(C)}$ is strictly monotone decreasing in $\tr \xi(P) \xi(R)$ and $K_\alpha^*\ler{R || C}$ is strictly monotone decreasing in $\tr P R.$
\end{proof}

Let us now recall Wigner's famous theorem on the structure of quantum mechanical symmetry transformations.
It states that any bijection of $\cP_1(\cH)$ onto itself which preserves transition probabilities (i.e., preserves the trace of the products of rank-one projections) is necessarily implemented by a unitary or an antiunitary operator. Therefore, we get that
\begin{equation*} \label{p1en}
\xi_{|\cP_1(\cH)}(R)=U R U^* \qquad  \ler{R \in \cP_1(\cH)}
\end{equation*}
for some unitary or antiunitary operator $U$ acting on $\cH.$
\par
We intend to show that $\xi(A)=UA U^*$ holds for any $A \in \cM(\cH).$ We mention that the core idea of the proof of this step appeared in \cite{lm08}, and that technique was further developed in \cite{dv16a}.
Let us define the map $\xi': \cS(\cH) \rightarrow \cS(\cH)$ by
$$
\xi'(A):=U^*\xi(A)U \qquad \ler{A \in \cS(\cH)}.
$$
By the assumptions, $\xi'$ has the same properties as $\xi$ plus it has the additional property that it acts identically on the set $\cP_1(\cH).$
Therefore,
\be \label{rezol}
K_\alpha\ler{P||\xi'(A)}=K_\alpha\ler{\xi'(P)||\xi'(A)}= K_\alpha\ler{P||A}
\ee
holds for any $A \in \cM(\cH)$ and for any $P \in \cP_1(\cH).$
Considering the equation \eqref{minertek}, it is clear by \eqref{rezol} that that the greatest eigenvalues of $A$ and $\xi'(A)$ coincide and, by \eqref{minelv}, it is also clear that the eigenprojections corresponding to the greatest eigenvalues coincide, too.
\par
The formula \eqref{szamszab2} shows that, similarly to the equations \eqref{minertek} and \eqref{minelv}, the following holds (here we use the notation of \eqref{ddef}, $A$ being in the place of $D$):
\begin{equation*} \label{minertek2}
\mathrm{min}\left\{K_\alpha^*\ler{X || A} \middle| X \in \cP_1(\cH), X P_1=0\right\}=\frac{1}{\lambda_2}
\end{equation*}
and
\begin{equation*} \label{minelv2}
\begin{gathered}
K_\alpha^*(R||A)=\mathrm{min}\left\{K_\alpha^*\ler{X || A} \middle| X \in \cP_1(\cH), X P_1=0\right\}\\
\Longleftrightarrow RP_1=0, R \leq P_2.
\end{gathered}
\end{equation*}
By \eqref{rezol} this means that the second greatest eigenvalues of $A$ and $\xi'(A)$ coincide, and so do the corresponding eigenprojections.
Continuing this process, after finitely many steps we get that $\xi'(A)=A.$ Therefore, $\xi'$ acts as the identity on the set of nonsingular densities, as well.
This means that $\xi(A)= UA U^*$ for any $A\in \mathcal M(\mc{H})$ as asserted in the beginning of this subsection.

\subsection{Proof of Theorem \ref{fo} --- part three} \label{3fej}

We are now in  a position to complete the proof of Theorem \ref{fo}. In what follows let $\phi$, $\psi$ be as in Subsection \ref{1fej}. Introduce the notation
$$
\mc{L}^+(\cH)_\lambda :=\left\{A \in \mc{L}^+(\cH): \, \tr A=\lambda\right\}.
$$
Observe that $\mc{L}^+(\cH)_1$ equals the state space which is denoted by $\cS(\cH).$ Furthermore, observe that by the trace preserving property given in Claim \ref{psitrace}, $\psi$ restricted to $\mc{L}^+(\cH)_\lambda$ is a bijection of that set onto itself. In particular, $\psi$ restricted to $\cS(\cH)$ is a bijection from $\cS(\cH)$ onto itself.
Straightforward computations show that the $\chi_\alpha^2$-divergence is homogeneous, that is,
$$
K_\alpha\ler{\lambda A || \lambda B}=\lambda K_\alpha\ler{A || B} \qquad \ler{A, B \in \mc{L}^+(\cH), \, \lambda \in [0, \infty)}.
$$
For any $\lambda \in (0,\infty),$ let us define a map $\psi_\lambda$ in the following way:
$$
\psi_\lambda: \cS(\cH) \rightarrow \cS(\cH), \quad A \mapsto \psi_\lambda(A):=\frac{1}{\lambda}\psi\ler{\lambda A}.
$$
The map $\psi_\lambda$ satisfies \eqref{E:M3} because
$$
K_\alpha\ler{\psi_\lambda(A)||\psi_\lambda(B)}
=K_\alpha\ler{\frac{1}{\lambda}\psi\ler{\lambda A}||\frac{1}{\lambda}\psi\ler{\lambda B}}
$$
$$
=\frac{1}{\lambda} K_\alpha\ler{\psi\ler{\lambda A}||\psi\ler{\lambda B}}
=\frac{1}{\lambda}K_\alpha\ler{\lambda A||\lambda B}=K_\alpha\ler{A||B}
$$
holds for any $A \in \cS(\cH), B \in \cM(\cH)$ and $\lambda \in (0, \infty).$
Moreover, by the bijectivity and the trace-preserving property of $\psi$, the map $\psi_\lambda$ is a bijection on $\cS(\cH).$ Moreover, it acts bijectively on $\cM(\cH),$ because
$$
\psi_\lambda(A)=\frac{1}{\lambda}\psi\ler{\lambda A}=
\frac{1}{\lambda}\phi\ler{\lambda A}
$$
holds for any invertible density $A$ and $\lambda \in (0, \infty),$ and $\phi$ is a trace-preserving bijection on $\mc{L}^{++}(\cH).$
So, the results in Section \ref{2fej} apply and for any $\lambda \in (0, \infty)$ we have that there is a unitary or an antiunitary operator $U_\lambda$ on $\mc{H}$ such that
$$
\phi(\lambda A)= U_\lambda (\lambda A) U_\lambda^*
$$
holds for any nonsingular density operator $A$ on $\mc{H}$. We need to show that $U_\lambda$ does not depend essentially on the parameter $\lambda$ meaning that all $U_\lambda$'s induce the same similarity transformation.
In order to verify this, fix positive real numbers $\lambda,\mu$.
Choose $A,B\in \cM(\cH)$. We have
\[
K_\alpha(\lambda A || \mu B)=K_\alpha\ler{\phi(\lambda A)||\phi(\mu B)}= K_\alpha \ler{U_\lambda (\lambda A) U_\lambda^* || U_\mu (\mu B) U_\mu^*}
\]
from which we easily deduce that
\[
\tr B^{-\alpha} AB^{\alpha -1}A=
\tr U_\mu B^{-\alpha} U_\mu^* U_\lambda A U_\lambda^* U_\mu B^{\alpha -1}U_\mu^* U_\lambda A U_\lambda^*
\]
holds for any $A,B\in \cM(\cH)$. Denoting $V=U_\mu^*U_\lambda $ we have
\[
\tr B^{-\alpha} AB^{\alpha -1}A=
\tr B^{-\alpha} V A V^* B^{\alpha -1}V A V^*
\]
for all $A,B\in \cM(\cH)$. Fix $B \in \cM(\cH)$. Then, first for all $A\in \cM(\cH)$ and then for all $A\in \mc{L}^{++}(\mc {H})$ and finally for all $A\in \mc{L}^{+}(\mc {H})$ we have
\[
\tr B^{-\alpha} AB^{\alpha -1}A=
\tr (V^*B^{-\alpha} V) A (V^* B^{\alpha -1}V) A. 
\]
Linearizing this equality, i.e., writing $A+A'$ in the place of $A$ we infer that the equality
\[
\begin{gathered}
\tr B^{-\alpha} AB^{\alpha -1}A' + \tr B^{-\alpha} A'B^{\alpha -1}A\\=
\tr (V^*B^{-\alpha} V) A (V^* B^{\alpha -1}V) A' +  \tr (V^*B^{-\alpha} V) A' (V^* B^{\alpha -1}V) A
\end{gathered}
\]
is valid for any $A,A'\in \mc{L}^{+}(\mc {H})$. We can rewrite this in the following way:
\[
\begin{gathered}
\tr \ler{ B^{-\alpha} AB^{\alpha -1} + B^{\alpha -1}AB^{-\alpha}}A'\\=
\tr \ler{(V^*B^{-\alpha} V) A (V^* B^{\alpha -1}V) +  (V^* B^{\alpha -1}V)  A(V^*B^{-\alpha} V)} A'
\end{gathered}
\]
for any $A,A'\in \mc{L}^{+}(\mc {H})$ and then for any $A,A'\in \mc{L}(\mc {H})$, too (every operator is a linear combination of positive semidefinite ones).  It easily follows that
\[
\begin{gathered}
B^{-\alpha} AB^{\alpha -1} +  B^{\alpha -1}AB^{-\alpha}\\=
(V^*B^{-\alpha} V) A (V^* B^{\alpha -1}V) +  (V^* B^{\alpha -1}V)  A(V^*B^{-\alpha} V)
\end{gathered}
\]
holds for any $A\in \mc{L}(\mc {H})$ and $B\in \cM(\cH)$. It is easy to see that, plugging $B^{-1}/\tr B^{-1}$ into the place of $B$, we next have
\[
B^{\alpha} AB^{1-\alpha } +  B^{1-\alpha}AB^{\alpha}=
(V^*B^{\alpha} V) A (V^* B^{1-\alpha}V) +  (V^* B^{1-\alpha }V)  A(V^*B^{\alpha} V)
\]
for any $A\in \mc{L}(\mc {H})$ and for any $B\in \mc{L}^{++}(\mc {H})$ and then for any $B\in \mc{L}^{+}(\mc {H})$, too.
Assume $0<\alpha<1$. Then it follows that for any projection $P$ on $\mc{H}$ we have
\[
2PAP=2(V^*PV) A(V^*PV) \qquad ( A\in \mc{L}(\mc {H})).
\]
This easily implies that $P=V^*PV$ for all projections $P$ on $\mc{H}$. This further gives that $V$ equals the identity multiplied by a complex number of modulus 1. It follows that $U_\lambda, U_\mu$ are linearly dependent for any $\lambda,\mu$ and hence they induce the same unitary or antiunitary similarity transformation. The argument is similar but simpler in the case where $\alpha$ is either 0 or 1. Consequently, we have a unitary or antiunitary operator $U$ on $\mc{H}$ such that
\[
\phi(A)=UAU^*\qquad ( A\in \mc{L}^{++}(\mc {H})).
\]
This completes the proof of our main result Theorem \ref{fo}.

\subsection{The sketches of the proofs of Propositions \ref{semidef} and \ref{density}} \label{4fej}

This subsection is devoted to give the sketches of the proofs of our results concerning bijective maps preserving the $\chi^2_\alpha$-divergence on the cone of all positive semidefinite operators or on the state space.

First, we consider Proposition \ref{semidef}. Let $\phi$ be 
the map given there.
We observe that a positive semidefinite operator $B$ on $\mc{H}$ is nonsingular if and only if we have $K_\alpha(A||B) <\infty $ for every $A\in\mc{L}^{+}(\cH)$.
Thus we infer that
$$
\phi\ler{\mc{L}^{++}(\cH)} = \mc{L}^{++}(\cH).
$$
Clearly, the restriction $\phi_{|\mc{L}^{++}(\cH)}$ satisfies the conditions of Theorem \ref{fo}, hence we have a unitary or an antiunitary operator $U: \cH \rightarrow \cH$ such that 
$$
\phi(A)=UAU^* \qquad \ler{A \in \mc{L}^{++}(\cH)}.
$$
Next, we observe that using the same argument as in the proof of Claim \ref{kiterj}, we can show that for $A\in\mc{L}^{+}(\cH)$, $\{A_n\}_{n=1}^\infty \subset \mc{L}^{++}(\cH)$ with $\lim_{n\to\infty} A_n = A$, we have that $\{\phi(A_n)\}_{n\in \mathbb N}$ is convergent. Let $L=\lim_{n\to\infty}\phi(A_n)$. To see that $L=\phi(A)$, for any $B\in \mc{L}^{++}(\mc{H})$ we compute
\begin{equation*}
\begin{gathered}
K_\alpha(L || \phi(B))=\lim_{n\to\infty} K_\alpha (\phi(A_n)|| \phi(B))\\=\lim_{n\to\infty} K_\alpha(A_n||B)=K_\alpha( A||B)=K_\alpha (\phi(A)||\phi(B))
\end{gathered}
\end{equation*}
which, by Lemma \ref{L:M}, implies that $\phi(A)=L$. 
(Note that this does not give the continuity of $\phi$ on $\mc{L}^{+}(\cH)$.)
Hence we obtain
$$
\phi(A)=\lim_{n\to\infty} \phi(A_n)= \lim_{n\to\infty} UA_nU^*=UAU^* 
$$
for any $A \in \mc{L}^+(\cH)$
with some given unitary or antiunitary operator $U$ on $\mc{H}$. This proves the statement in Proposition \ref{semidef}.

As for Proposition \ref{density}, let $\phi$ be the map given there. As above, we can easily deduce that
$$
\phi(\cM(\cH)) = \cM(\cH)
$$
and that $\phi$ satisfies \eqref{E:M3}.
Therefore, by the results of Subsection \ref{2fej}, we have a unitary or an antiunitary transformation $U: \cH \rightarrow \cH$ with
$$
\phi(A)=UAU^* \qquad \ler{A \in \cM(\cH)}.
$$
The proof can now be completed in a way very similar to the last part of the proof of Proposition \ref{semidef}.

\section{Conclusion, open problems}

Above we have proven that any bijective map on any of the convex sets $\mc{L}^{++}(\cH), \mc{L}^{+}(\cH),  \mc{S}(\cH)$ which preserve the $\chi^2_\alpha$-divergence is a unitary or an antiunitary similarity transformation. This gives the somewhat surprising conclusion that although the quantity $K_\alpha(.||.)$ is highly nonlinear in its variables, the bijective maps which preserve it are linear, more accurately, affine automorphisms of the underlying convex sets.

We finish the paper with two very natural and exciting questions to which we do not have answers and hence we leave them as open problems. First, we ask if the bijectivity assumptions in our results above can be relaxed. Second, what is the structure of those bijective maps on the sets $\mc{L}^{++}(\cH)$, $\mc{L}^+(\cH)$, $\mc{S}(\cH)$ which preserve a general $\chi^2$-divergence given in (10) on page 122201-3 in \cite{Ruskaietal}. The arguments we have presented in this paper may convince the reader that those questions are most probably difficult and hence rather challenging and may therefore be the targets of further investigations.


\begin{thebibliography}{99}


\bibitem{carlen} E. Carlen, \emph{Trace inequalities and quantum entropy: an introductory course,} Contemp. Math. {\bf 529} (2010), 73--140.



\bibitem{gyuri} Gy. P. Geh\'er, \emph{An elementary proof for the non-bijective version of Wigner’s theorem,} Phys. Lett. A {\bf 378} (2014), 2054--2057.

\bibitem{Hansen}
F. Hansen, \emph{Convexity of quantum $\chi^2$-divergence,} Proc. Natl. Acad. Sci. USA \textbf{108} (2011), 10078--10080.

\bibitem{HMPB}
F. Hiai, M. Mosonyi, D. Petz and C. Bény, \emph{Quantum $f$-divergences and error correction,} Rev. Math. Phys. \textbf{23} (2011), 691--747.

\bibitem{HP}
F. Hiai and D. Petz, 
\emph{From quasi-entropy to various quantum information quantities,}
Publ. Res. Inst. Math. Sci. \textbf{48} (2012), 525--542. 

\bibitem{Jencova}
A. Jen\v cov\'a, 
\emph{Reversibility conditions for quantum operations,} 
Rev. Math. Phys. \textbf{24} (2012), 1250016, 26 pp. 



\bibitem{MB}
L. Moln\'ar,
\emph{Selected Preserver Problems on Algebraic Structures of Linear
Operators and on Function Spaces,} Lecture Notes in Mathematics, Vol. 1895, p. 236, Springer, 2007.

\bibitem{lm08} L. Moln\'ar, \emph{Maps on states preserving the relative entropy,} J. Math. Phys. {\bf 49,} 032114 (2008).


\bibitem{ML16g} 
L. Moln\'ar,
\emph{Maps on the positive definite cone of a $C^*$-algebra preserving certain quasi-entropies,} J. Math. Anal. Appl. \textbf{447} (2017), 206--221.

\bibitem{molnagy} 
L. Moln\'ar and G. Nagy, \emph{Isometries and relative entropy preserving maps on density operators,} Linear Multilinear Algebra {\bf 60,} (2012), 93--108.

\bibitem{mnsz13} 
L. Moln\'ar, G. Nagy and P. Szokol, \emph{Maps on density operators preserving quantum $f$-divergences,} Quantum Inf. Process. {\bf 12} (2013), 2309--2323.

\bibitem{mpv15} 
L. Moln\'ar, J. Pitrik and D. Virosztek, \emph{Maps on positive definite matrices preserving Bregman and Jensen divergences,} Linear Algebra Appl. {\bf 495} (2016), 174--189.

\bibitem{molszok}
L. Moln\'ar and P. Szokol,
\emph{Maps on states preserving the relative entropy II,}
Linear Algebra Appl. \textbf{432} (2010), 3343--3350.

\bibitem{PG}
D. Petz and C. Ghinea, 
\emph{Introduction to quantum Fisher information,} in Quantum probability and related topics, 261--281,
QP-PQ: Quantum Probab. White Noise Anal., 27, World Sci. Publ., Hackensack, NJ, 2011. 

\bibitem{Ruskaietal}
K. Temme, M.J. Kastoryano, M.B. Ruskai, M.M. Wolf and F. Verstraete,
\emph{The $\chi^2$-divergence and mixing times of quantum Markov processes,}
J. Math. Phys. \textbf{51} (2010), 122201, 19 pp. 

\bibitem{Temmeetal}
K. Temme and F. Verstraete, 
\emph{Quantum chi-squared and goodness of fit testing,}
J. Math. Phys. \textbf{56} (2015), 012202, 18 pp. 

\bibitem{dv16a} 
D. Virosztek, \emph{Maps on quantum states preserving Bregman and Jensen divergences,}  Lett. Math. Phys. \textbf{106} (2016), 1217--1234.

\bibitem{dv16b} D. Virosztek, \emph{Quantum f-divergence preserving maps on positive semidefinite operators acting on
finite dimensional Hilbert spaces,} Linear Algebra Appl. {\bf 501}  (2016), 242--253.
\end{thebibliography}
\end{document}